\documentclass[a4paper,numberwithinsect,UKenglish,cleveref,autoref,thm-restate]{lipics-v2019}

\usepackage[utf8]{inputenc}

\usepackage{thmtools}
\usepackage{extarrows}
\usepackage{microtype}
\usepackage{mathtools}
\usepackage{amsmath}
\usepackage{amsthm}

\usepackage{amssymb}
\usepackage{bbm}
\usepackage{stmaryrd}
\newcommand{\ldbrack}{\llbracket}
\newcommand{\rdbrack}{\rrbracket}
\usepackage{tikz}
\usetikzlibrary{patterns,arrows,arrows.meta,calc}
\pgfdeclarelayer{background}
\pgfdeclarelayer{foreground}
\pgfsetlayers{background, main, foreground}
\usepackage[textwidth=2cm,textsize=small]{todonotes}
\usepackage{xspace}

\usepackage{wrapfig}
\usepackage{paralist}
\usepackage[all]{xy}

\usepackage{scalefnt}
    
\newcommand{\confof}[1]{c(#1)}
\newcommand{\closure}[2]{\Pi_{#1}(#2)}

\newcommand{\XX}{\mathcal X}
\newcommand{\YY}{\mathcal Y}
\newcommand{\ZZ}{\mathcal Z}

\newcommand{\pair}[2]{\tuple {#1, #2}}
\newcommand{\nowof}[1]{\text{\sc now}(#1)}

\newcommand{\orbit}[2]{\text{\sc orbit}_{#1}(#2)} 

\newcommand{\Conf}[1]{\semd {#1}}
\newcommand{\clockval}[1]{\mu(\X)}
\newcommand{\macval}[1]{\text{\sc Val}(#1)}
\newcommand{\assigntoval}[2]{{#2} - {#1}}
\newcommand{\now}{t_0} 

\newcommand{\para}[1]{\subparagraph{#1.}}
\newcommand{\succe}[2]{\text{\sc succ}_{#1}(#2)}

\newcommand{\HyperAckermann}{{\sc Hyper\-Acker\-mann}}

\newcommand{\ignore}[1]{}

\newcommand{\goesto}[1]{\xrightarrow{#1}}

\newcommand{\lang}[1]{L(#1)}
\newcommand{\langa}[2]{L_{#1}(#2)}
\newcommand{\langtsa}[2]{\langa{#1}{#2}} 
\newcommand{\langts}[1]{\lang{#1}}         

\newcommand{\true}{\mathbf{true}}
\newcommand{\false}{\mathbf{false}}
\newcommand{\N}{\mathbb N}
\newcommand{\Z}{\mathbb Z}

\newcommand{\Rpos}{{\mathbb R}_{\geq 0}}
\newcommand{\R}{\Rpos}
\newcommand{\Rreal}{{\mathbb R}}

\newcommand{\I}{\mathbb I}

\newcommand{\X}{\mathtt X}
\newcommand{\Y}{\mathtt Y}

\newcommand{\x}{\mathtt x}
\newcommand{\y}{\mathtt y}

\renewcommand{\L}{\mathtt L}
\renewcommand{\I}{\mathtt I}
\newcommand{\F}{\mathtt F} 
\newcommand{\op}{\mathtt {op}}

\newcommand{\automaton}{\tuple{\Sigma, \L, \X, \I, \F, \Delta}}

\newcommand{\extend}[3]{#3[#1 \mapsto #2]}

\newcommand{\timedwords}[1]{\mathbb{T}(#1)}
\newcommand{\timedwordsafter}[2]{\mathbb{T}_{\geq{#2}}(#1)}

\newcommand{\transition}[5]{\tuple{#1, #2, #3, #4, #5}}

\newcommand{\regions}[2]{\textsf{Reg}(#1,#2)}
\newcommand{\reg}[2]{\regions{#1}{#2}} 

\newcommand{\fract}[1]{\mathsf{fract}({#1})}

\newcommand{\DTA}{\textsf{\sc dta}\xspace}
\newcommand{\mDTA}[1]{\ensuremath{{\textsf{\sc dta}_{\textnormal{\_},#1}}}\xspace}
\newcommand{\kDTA}[1]{\ensuremath{{\textsf{\sc dta}_{#1}}}\xspace}
\newcommand{\kmDTA}[2]{\ensuremath{{\textsf{\sc dta}_{#1,#2}}}\xspace}

\newcommand{\NTA}{\textsf{\sc nta}\xspace}
\newcommand{\mNTA}[1]{\ensuremath{{\textsf{\sc nta}_{\textnormal{\_},#1}}}\xspace}
\newcommand{\kNTA}[1]{\ensuremath{{\textsf{\sc nta}_{#1}}}\xspace}
\newcommand{\kNTAe}[1]{\ensuremath{{\textsf{\sc nta}_{#1}^\varepsilon}}\xspace}
\newcommand{\kmNTA}[2]{\ensuremath{{\textsf{\sc nta}_{#1,#2}}}\xspace}

\newcommand{\semd}[1]{\left\ldbrack#1\right\rdbrack}
\newcommand{\semlog}[1]{\left\ldbrack#1\right\rdbrack}
\newcommand{\prettyexists}[2]{\exists #1\cdot #2}

\newcommand{\sep}{\ | \ }
\newcommand{\tuple}[1]{\left(#1\right)}

\newcommand{\set}[1]{\left\{ #1 \right\}}
\newcommand{\setof}[2]{\set{#1 \; \middle| \; #2}}

\newcommand{\LCM}{\textsf{LCM}\xspace}
\newcommand{\kLCM}[1]{$#1$-\textsf{LCM}\xspace}
\newcommand{\incr}[1]{#1\,\texttt{++}}
\newcommand{\decr}[1]{#1\,\texttt{--}}
\newcommand{\ztest}[1]{#1\stackrel ? {\texttt=} 0}

\newcommand{\reachset}[1]{\text{Reach}(#1)}

\newif\ifstartedinmathmode
\renewcommand*{\st}{
  \relax\ifmmode\startedinmathmodetrue\else\startedinmathmodefalse\fi
  \ifstartedinmathmode{\;\cdot\;}\else{s.t.~}\fi%
}

\newcommand{\wrt}{w.r.t.~}

\newcommand{\pspace}{{\sc PSpace}\xspace}

\newcommand{\nlogspace}{{\sc NLogSpace}\xspace}

\newtheorem{fact}[theorem]{Fact}

\newcommand{\decision}[3]{\medskip\noindent {\sc #1}. \newline {\bf Input: } #2 \newline {\bf Output: } #3\medskip}

\crefname{claim}{Claim}{Claims}
\Crefname{claim}{Claim}{Claims}

\crefname{lemma}{Lemma}{Lemmas}
\Crefname{Lemma}{Lemma}{Lemmas}

\crefname{theorem}{Theorem}{Theorems}
\Crefname{Theorem}{Theorem}{Theorems}

\crefname{fact}{Fact}{Facts}
\Crefname{fact}{Fact}{Facts}

\title{Determinisability of one-clock timed automata}

\author
	{Lorenzo Clemente}
	{University of Warsaw, Poland}
	{clementelorenzo@gmail.com}
	{https://orcid.org/0000-0003-0578-9103}
	{Partially supported by the Polish NCN grant 2017/26/D/ST6/00201.}
\author
	{Sławomir Lasota}
	{University of Warsaw, Poland}
	{sl@mimuw.edu.pl}
	{{https://orcid.org/0000-0001-8674-4470}}
	{Partially supported by the Polish NCN grant 2019/35/B/ST6/02322 and by the ERC grant LIPA, agreement no.~683080.}
\author
	{Radosław Piórkowski}
	{University of Warsaw, Poland}
	{r.piorkowski@mimuw.edu.pl}
	{https://orcid.org/0000-0002-9643-182X}
	{Partially supported by the Polish NCN grant 2017/27/B/ST6/02093.}

\authorrunning{L.~Clemente, S.~Lasota, and R.~Piórkowski} 
\Copyright{Lorenzo Clemente and Sławomir Lasota and Radosław Piórkowski} 



\category{} 

\relatedversion{}

\supplement{}

\acknowledgements{We thank S. Krishna for fruitful discussions and the anonymous reviewers for their constructive comments.}

\nolinenumbers 

\EventEditors{Igor Konnov and Laura Kov\'{a}cs}
\EventNoEds{2}
\EventLongTitle{31st International Conference on Concurrency Theory (CONCUR 2020)}
\EventShortTitle{CONCUR 2020}
\EventAcronym{CONCUR}
\EventYear{2020}
\EventDate{September 1--4, 2020}
\EventLocation{Vienna, Austria}
\EventLogo{}
\SeriesVolume{2017}
\ArticleNo{38}

\begin{document}

\maketitle

\begin{abstract}
	The deterministic membership problem for timed automata asks whether
	the timed language recognised by a nondeterministic timed automaton 
	can be recognised by a deterministic timed automaton.
	We show that the problem is decidable when the input automaton is a one-clock nondeterministic timed automaton without epsilon transitions
	and the number of clocks of the deterministic timed automaton is fixed.
	We show that the problem in all the other cases is undecidable, i.e.,
	when either 1) the input nondeterministic timed automaton has two clocks or more,
	or 2) it uses epsilon transitions,
	or 3) the number of clocks of the output deterministic automaton is not fixed.
\end{abstract}

\newpage 


\section{Introduction}

Nondeterministic timed automata (\NTA) are one of the most widespread model of real-time reactive systems.
They are an extension of finite automata
with real-valued clocks which can be reset and compared by inequality constraints.
The nonemptiness problem for \NTA is decidable and in fact \pspace-complete,
as shown by Alur and Dill in their landmark paper \cite{AD94}.
As a testimony to the importance of the model,
the authors received the 2016 Church Award \cite{church:award}
for the invention of timed automata.
This paved the way to the automatic verification of timed systems,
leading to mature tools such as UPPAAL \cite{Behrmann:2006:UPPAAL4},
UPPAAL Tiga (timed games) \cite{CassezDavidFleuryLarsenLime:CONCUR:2005},
and PRISM (probabilistic timed automata) \cite{KwiatkowskaNormanParker:CAV:2011}.
%
The reachability problem is still a very active research area to these days
\cite{FearnleyJurdziski:IC:2015,
HerbreteauSrivathsanWalukiewicz:IC:2016,
AkshayGastinKrishna:LMCS:2018,
GastinMukherjeeSrivathsan:CONCUR:2018,
GastinMukherjeeSrivathsan:CAV:2019,
GovindHerbreteauSrivathsanWalukiewicz:CONCUR:2019},
as well as expressive generalisations thereof,
such as the binary reachability problem
\cite{ComonJurski:TA:1999,
Dima:Reach:TA:LICS02,
KrcalPelanek:TM:FSTTCS:2005,
FranzleQuaasShirmohammadiWorrell:IPL:2020}.

Deterministic timed automata (\DTA) form a strict subclass of \NTA
where the next configuration is uniquely determined from the current one and the timed input symbol.
The class of \DTA enjoys stronger properties than \NTA,
such as decidable universality and inclusion problems
and closure under complementation \cite{AD94}.
Moreover, the more restrictive nature of \DTA is necessary in several applications of timed automata,
such as test generation \cite{NielsenSkou:STTT:2003},
fault diagnosis \cite{BouyerChevalierDSouza:FOSSACS:2005},
and learning \cite{VerwerWeerdtWitteveen:Benelearn:2007,TapplerAichernigLarsenLorber:FOMATS:2019},
winning conditions in timed games
\cite{AsarinMaler:HSCC:1999,JurdzinskiTrivedi:ICALP:2007,BrihayeHenzingerPrabhuRaskin:ICALP:2007},
and in a notion of recognisability of timed languages \cite{Maler:Pnueli:FOSSACS:04}.
For these reasons,
and for the more general quest of understanding the nature of the expressive power of nondeterminism in timed automata,
many researchers have focused on defining determinisable classes of timed automata,
such as strongly non-zeno \NTA \cite{AsarianMalerPnueliSifakis:SSSC:1998},
event-clock \NTA \cite{AlurFixHenzinger:TCS:1999},
and \NTA with integer-resets \cite{SumanPandyaKrishnaManasa:2008}.
The classes above are not exhaustive,
in the sense that there are \NTA recognising deterministic timed languages
not falling into any of the classes above.
%
%

Another remarkable subclass of \NTA is obtained by requiring the presence of just one clock (without epsilon transitions).
The resulting class of \kNTA 1 is incomparable with \DTA:
For instance, \kNTA 1 are not closed under complement (unlike \DTA)
and there are very simple \DTA languages which are not recognisable by any \kNTA 1.
Nonetheless, \kNTA 1, like \DTA, have decidable inclusion, equivalence, and universality problems \cite{OW04,LasotaWalukiewicz:ATA:ACM08},
albeit the complexity is non-primitive recursive \cite[Corollary 4.2]{LasotaWalukiewicz:ATA:ACM08}
(see also \cite[Theorem 7.2]{OuaknineWorrel:LMCS:2007} for an analogous lower bound for the satisfiability problem of metric temporal logic).
Moreover, the non-emptiness problem for \kNTA 1 is \nlogspace-complete
(vs.~\pspace-complete for unrestricted \NTA and \DTA, already with two clocks \cite{FearnleyJurdziski:IC:2015}),
and computing the binary reachability relation is simpler when there is only one clock than in the general case
\cite{ClementeHofmanTotzke:CONCUR:2019}.

\para{The deterministic membership problem}

The \emph{\DTA membership problem} asks, given an \NTA,
whether there exists a \DTA recognising the same language.
There are two natural variants of this problem,
which are obtained by restricting the resources available to the sought \DTA.
Let $k \in \N$ be a bound on the number of clocks,
and let $m \in \N$ be a bound on the maximal absolute value of numerical constants.
The \emph{\kDTA k}
and \emph{\kmDTA k m membership problems}
are the restriction of the problem above
where the \DTA is required to have at most $k$ clocks,
resp., at most $k$ clocks and absolute value of maximal constant bounded by $m$.
Notice that we do not bound the number of control locations of the \DTA,
which makes the problem non-trivial.

Since regular languages are deterministic,
the \kDTA k membership problem can be seen as a quantitative generalisation of the regularity problem.
For instance, the $\kDTA 0$ membership problem is exactly the regularity problem
since a timed automaton with no clocks is the same as a finite automaton.
We remark that the regularity problem is usually undecidable
for nondeterministic models of computation generalising finite automata,
e.g., context-free grammars/pushdown automata \cite[Theorem 6.6.6]{Shallit:Book:2008},
labelled Petri nets under reachability semantics \cite{ValkVidal-Naquet:Petri:Regular:1981}, Parikh automata \cite{CadilhacFinkelMcKenzie:NCMA:2011}, etc.
One way to obtain decidability is to either restrict the input model to be deterministic
(e.g., \cite{Valiant:Regularity:DPDA:JACM:1975,ValkVidal-Naquet:Petri:Regular:1981,BaranyLodingSerre:STACS:2006}),
or to consider finer notions of equivalence,
such as bisimulation (e.g., \cite{ParysGoller:LICS:2020}).

This negative situation is generally confirmed for timed automata.
For every number of clocks $k\in\N$ and maximal constant $m$,
the \DTA, \kDTA k, and \kmDTA k m membership problems are known to be undecidable
when the input \NTA has $\geq 2$ clocks,
and for 1-clock \NTA with epsilon transitions \cite{Finkel:FORMATS:2006,Tripakis:IPL:2006}.
To the best of our knowledge,
the deterministic membership problem was not studied before when the input automaton is \kNTA 1 without epsilon transitions.

\para{Contributions}

We complete the study of the decidability border for the deterministic membership problem
initiated in \cite{Finkel:FORMATS:2006,Tripakis:IPL:2006}.
Our main result is the following.
\begin{restatable}{theorem}{thmkDTAmemb}
    \label{thm:kDTA:memb}
    The \kDTA{k} membership and the \kmDTA k m membership problems
    are decidable for \kNTA{1} languages.
\end{restatable}
Our decidability result contrasts starkly
with the abdundance of undecidability results for the regularity problem.
We establish decidability by showing that if a \kmNTA k m recognises a \kDTA k language,
then in fact it recognises a \kmDTA k m language
and moreover there is a computable bound on the number of control locations of the deterministic acceptor (c.f.~\cref{thm:k-DTA-char}).
This provides a decision procedure
since there are finitely many \DTA once the number of clocks, the maximal constant, and the number of control locations are fixed.


In our technical analysis we find it convenient to introduce the so called
\emph{always resetting} subclass of \kNTA k. 
These automata are required to reset at least one clock at every transition
and are thus of expressive power intermediate between \kNTA {k-1} and \kNTA k.
Always resetting \kNTA 2 are strictly more expressive than \kNTA 1: For instance, the language of timed words of the form
$(a, t_0) (a, t_1) (a, t_2)$ \st $t_2 - t_0 > 2$ and $t_2 - t_1 < 1$ can be recognised by an always resetting \kNTA 2
but by no \kNTA 1.
Despite their increased expressive power,
always resetting \kNTA 2
still have a decidable universality problem (the well-quasi order approach of \cite{OW04} goes through),
which is not the case for \kNTA 2.
Thanks to this restricted form,
we are able to provide in \cref{thm:k-DTA-char} an elegant characterisation of those \kNTA 1 languages
which are recognised by an always resetting \kDTA k.

We complement the decidability result above by showing that the problem becomes undecidable
if we do not restrict the number of clocks of the \DTA.
\begin{restatable}{theorem}{thmUndec}
    \label{thm:undecidability}
    The \DTA and \mDTA m ($m > 0$) membership problems are undecidable for \kNTA 1 without epsilon transitions.
\end{restatable}
Finally, by refining the analysis of \cite{Finkel:FORMATS:2006},
we show that the \kDTA{k} and \kmDTA k m membership problems for \kNTA 1
are non-primitive recursive.
\begin{restatable}{theorem}{thmHardness}
    \label{thm:hardness}
    The \kDTA{k} and \kmDTA k m membership problems are \HyperAckermann-hard for \kNTA 1.
\end{restatable}

\para{Related research}

Many works addressed the construction of a \DTA equivalent to a given \NTA
(see \cite{BertrandStainerJeronKrichen:FMSD:2015} and references therein),
however since the general problem is undecidable,
one has to either sacrifice termination,
or consider deterministic under/over-approximations.
In a related line of work,
we have shown that the \emph{deterministic separability problem} is decidable for the full class of \NTA, when the number of clocks of the separator is given in the input \cite{ClementeLasotaPiorkowski:ICALP:2020}.
This contrasts with undecidability of the corresponding membership problem.
Decidability of the deterministic separability problem when the number of clocks of the separator is not provided remains a challenging open problem.


\section{Preliminaries}  \label{sec:prelim}

\para{Timed words and languages}

Fix a finite alphabet $\Sigma$.
Let $\Rreal$ and $\Rpos$ denote reals and nonnegative reals\footnote{Equivalently, nonnegative rationals may be considered
in place of reals.}, respectively.
%
A \emph{timed word} over $\Sigma$ is any sequence of the form
\begin{align} \label{eq:newtw}
	w \ = \ 
	(a_1, t_1)\, \dots \, (a_n, t_n)  \ \in \ (\Sigma \times\R)^*
\end{align}
%
which is \emph{monotonic},
in the sense that the timestamps $t_i$'s satisfy $0 \leq t_1 \leq t_2 \leq \dots \leq t_n$.
%
Let $\timedwords{\Sigma}$ be the set of all timed words over $\Sigma$,
and let $\timedwordsafter{\Sigma}{t}$ be, for $t\in \R$,
the set of timed words with $t_1 \geq t$. 
A \emph{timed language} is a subset of $\timedwords{\Sigma}$.


The concatenation $w \cdot v$ of two timed words $w$ and $v$
is defined only when the first time-stamp of $v$ is greater or equal than the last timestamp of $w$.
%
Using this partial operation, we define, 
for a timed word $w\in\timedwords{\Sigma}$ and a timed language $L\subseteq \timedwords\Sigma$, the left quotient
$w^{-1} L := \setof{v\in\timedwords\Sigma}{w \cdot v \in L}$.
Clearly $w^{-1} L \subseteq \timedwordsafter \Sigma {t_n}$. 

\para{Clock constraints and regions}
Let $\X = \set{\x_1, \dots, \x_k}$ be a finite set of clocks.
A \emph{clock valuation} is a function $\mu \in \Rpos^\X$
assigning a non-negative real number $\mu(\x)$ to every clock $\x \in \X$.
A \emph{clock constraint} is a quantifier-free formula of the form
\begin{align*}
    \varphi,\psi \ ::\equiv\ \true \sep \false \sep \x_i - \x_j \sim z \sep \x_i \sim z \sep \neg \varphi \sep \varphi \land \psi \sep \varphi \lor \psi,
\end{align*}
where ``$\sim$'' is a comparison operator in $\set{=, <, \leq, >, \geq}$
and $z \in \Z$.
A clock valuation $\mu$ satisfies a constraint $\varphi$, written $\mu \models \varphi$,
if interpreting each clock $\x_i$ by $\mu(\x_i)$ makes $\varphi$ a tautology.
%
%
An \emph{$k, m$-region} is a non-empty set of valuations $\semlog{\varphi}$
satisfied by a constraint $\varphi$ with $k$ clocks and absolute value of maximal constant bounded by $m$,
which is minimal \wrt set inclusion.
%
For instance, the clock constraint
$1 < \x_1 < 2 \;\wedge\; 4 < \x_2 < 5 \;\wedge\; \x_2 - \x_1 < 3$
defines a $2,5$-region consisting of an open triangle with nodes $(1, 4)$, $(2, 4)$ and $(2, 5)$.

\para{Timed automata}

A (nondeterministic) \emph{timed automaton} is a tuple $A = \automaton$,
where $\Sigma$ is a finite input alphabet,
$\L$ is a finite set of control locations,
$\X$ is a finite set of clocks,
$\I, \F \subseteq \L$ are the subsets of initial, resp., final, control locations,
and $\Delta$ is a finite set of transition rules of the form 
\begin{align} \label{eq:trans-rule}
	\transition{p}{a}{\varphi}{\Y}{q}
\end{align}
with $p, q \in \L$ control locations, $a \in \Sigma$,  
$\varphi$ a clock constraint to be tested,
and $\Y \subseteq \X$ the set of clocks to be reset.
We write \NTA for the class of all nondeterministic timed automata, \kNTA{k} when the number $k$ of clocks is fixed,
\mNTA{m} when the bound $m$ on constants is fixed,
and \kmNTA{k}{m} when both $k$ and $m$ are fixed.

An \mNTA m $A$ is \emph{always resetting} if every transition rule as in~\eqref{eq:trans-rule} resets some clock $\Y \neq \emptyset$,
and \emph{greedily resetting}
if, for every clock $\x$,
whenever $\varphi$ implies that $\x$ belongs to $\set{0, \ldots, m} \cup (m, \infty)$, then $\x \in \Y$.

\para{Reset-point semantics}

A \emph{configuration} of an \NTA $A$ is a tuple $\tuple{p, \mu, \now}$
consisting of a control location $p \in \L$,
a reset-point assignment $\mu \in \R^\X$, and a ``now'' timestamp $\now \in \R$
satisfying $\mu(\x) \leq \now$ for all clocks $\x\in \X$.
Intuitively, $\now$ is the last timestamp seen in the input
and, for every clock $\x$,
$\mu(\x)$ stores the timestamp of the last reset of $\x$.
A configuration is \emph{initial} if $p$ is so, $\now = 0$, and $\mu(\x) = 0$ for all clocks $\x$,
and it is \emph{final} if $p$ is so
(without any further restriction on $\mu$ or $\now$).
%
For a set of clocks $\Y \subseteq \X$ and a timestamp $u\in \R$, let $\extend \Y u \mu$
be the assignment which is $u$ on $\Y$ and agrees with $\mu$ on $\X \setminus \Y$.
An assignment $\mu$ together with $\now$ induces a clock valuation $\assigntoval \mu \now$ defined as 
$(\assigntoval \mu \now)(\x) = \now - \mu(\x)$ for all clocks $\x\in\X$.
Clock assignments and valuations have the same type $\R^\X$,
however we find it technically convenient to store assignments in configurations
and use the derived valuations to interpret the clock constraints.
Such reset-point semantics based on reset-point assignments
has already appeared in the literature on timed automata \cite{Fribourg:1998}
and it is the foundation of the related model of timed-register automata \cite{BL12}.

Every transition rule~\eqref{eq:trans-rule}
induces a \emph{transition} between configurations $\tuple {p, \mu, \now} \goesto {a,t} \tuple {q, \nu, t}$
labelled by $(a,t)\in\Sigma\times\R$ 
whenever 
%
    $t\geq \now$, 
    $\assigntoval \mu t \models \varphi$, and 
    $\nu = \extend \Y t \mu$.
%
\noindent
%
The \emph{timed transition system} induced by $A$
is $\tuple{\Conf A, \goesto {}, F}$,
where $\Conf A$ is the set of configurations,
${\goesto{}} \subseteq \Conf A \times \Sigma \times \Rpos \times \Conf A$
is as defined above, and
$F \subseteq \Conf A$ is the set of final configurations.
Since there is no danger of confusion,
we use $\semd A$ to denote either the timed transition system above, or its domain.
A \emph{run} of $A$ \emph{over} a timed word $w$ as in \eqref{eq:newtw}
\emph{starting} in configuration $\tuple {p, \mu, t_0}$
and \emph{ending} in configuration $\tuple {q, \nu, t_n}$
is a path $\rho$ in $\semd A$ 
of the form
%
$	\rho  =  \tuple {p, \mu, t_0} 
		\goesto {a_1,t_1} \ 
		\dots 
		\goesto {a_n, t_n} \ 
		\tuple{q, \nu, t_{n}}$. 
%
%
The run $\rho$ is accepting if its last configuration satisfies $\tuple {q, \nu, t_{n}}\in F$.
The language \emph{recognised} by configuration $(p, \mu, \now)$ is defined as: 
\begin{align*}
	\langtsa {\semd A} {p, \mu, \now} = \setof{w\in\timedwords{\Sigma}}{\semd A \text{ has an accepting run over } w \text{ starting in } \tuple{p, \mu, \now}}. 
\end{align*}
%
Clearly $\langtsa {\semd A} {p, \mu, \now} \subseteq \timedwordsafter \Sigma \now$.
We write $\langtsa {A} c$ instead of $\langtsa {\semd A} c$.
The language recognised by the automaton $A$ is
$\langts A = \bigcup_{c \text{ initial}} \langtsa A {c}$.
%
A configuration is \emph{reachable} if it is the ending configuration in a run starting in an initial configuration.
In an always resetting \mNTA m,
every reachable configuration $(p, \mu, \now)$ satisfies $\now \in \mu(\X)$,
and in a greedily resetting one,
1) $(p, \mu, \now)$ has \emph{$m$-bounded span},
in the sense that $\mu(\X) \subseteq (\now-m, \now]$,
and moreover 2) any two clocks $\x, \y$ with integer difference
$\mu(\x) - \mu(\y) \in \Z$ are actually equal $\mu(\x) = \mu(\y)$.
Condition 2) follows from the fact that if $\x, \y$ have integer difference and $\y$ was reset last,
then $\x$ was itself an integer when this happened,
and in fact they were both reset together in a greedily resetting automaton.

\para{Deterministic timed automata}
A timed automaton $A$ 
is \emph{deterministic} if it has exactly one initial location
and, for every two rules 
$\transition{p}{a}{\varphi}{\Y}{q}, \transition{p}{a'}{\varphi'}{\Y'}{q'} \in \Delta$,
if $a = a'$ and  $\semlog{\varphi \land \varphi'} \neq \emptyset$ then $\Y = \Y'$ and $q = q'$.
Hence $A$ has at most one run over every timed word $w$.
A \DTA can be easily transformed to a \emph{total} one, where for every location $p\in\L$ and $a\in\Sigma$,
the sets defined by clock constraints 
$\setof{\semlog \varphi}{\prettyexists{\Y, q}{\transition{p}{a}{\varphi}{\Y}{q} \in \Delta}}$
are a partition of $\Rpos^\X$.
Thus, a total \DTA has exactly one run over every timed word $w$.
%
We write \DTA for the class of deterministic timed automata,
and \kDTA{k}, \mDTA{m}, and \kmDTA{k}{m} for the respective subclasses thereof.
A timed language is called \NTA language, \DTA language, etc.,
if it is recognised by a timed automaton of the respective type.
%

\begin{example} \label{example:L1}
    Let $\Sigma = \set a$ be a unary alphabet.
	As an example of a timed language $L$ recognised by a \kNTA{1},
	but not by any \DTA,
    consider the set of non-negative timed words of the form
		$(a, t_1) \cdots (a, t_n)$
	where $t_n - t_i = 1$ for some $1\leq i < n$.
	The language $L$ is recognised by the \kNTA{1} $A = \automaton$
	with a single clock $\X = \set \x$ and three locations $\L = \set{p, q, r}$,
	of which $\I = \set p$ is initial and $\F = \set r$ is final, and transition rules
	\begin{align*} 
		& \transition{p}{a}{\true}{\emptyset}{p} \qquad 
		\transition{p}{a}{\true}{\set{\x}}{q} \qquad
		\transition{q}{a}{\x<1}{\emptyset}{q}  \qquad
		\transition{q}{a}{\x=1}{\emptyset}{r}.
	\end{align*}
	Intuitively, in $p$ the automaton waits until it guesses that the next input will be $(a, t_i)$,
	at which point it moves to $q$ by resetting the clock (and subsequently reading $a$).
	From $q$, the automaton can accept by going to $r$ only if exactly one time unit elapsed since $(a, t_i)$ was read.
	The language $L$ is not recognised by any \DTA
	since, intuitively, any deterministic acceptor needs to store unboundedly many timestamps $t_i$'s.
\end{example}

\para{Deterministic membership problems}

Let $\mathcal X$ be a subclass of \NTA.
We are interested in the following decision problem.
%


\decision{
$\mathcal X$ membership problem}
{A timed automaton $A\in$ \NTA.}
{Does there exist a 
$B\in \mathcal X$ \st $\lang A = \lang B$?}

In the rest of the paper, we study the decidability status of the $\mathcal X$ membership problem
where $\mathcal X$ ranges over \DTA, \kDTA k (for every fixed number of clocks $k$),
\mDTA m (for every maximal constant $m$),
and \kmDTA k m (when both clocks $k$ and maximal constant $m$ are fixed).
\Cref{example:L1} shows that there are \NTA languages
which cannot be accepted by any \DTA.
Moreover, there is no computable bound for the number of clocks $k$
which suffice to recognise a \kNTA 1 language by a \kDTA k (when such a number exists),
which follows from the following three observations:
1) the \DTA membership problem is undecidable for \kNTA 1
(\cref{thm:undecidability}),
2) the problem of deciding equivalence of a given \kNTA 1 to a given \DTA is decidable \cite{OW04}, and
3) if a \kmNTA 1 m is equivalent to some \kDTA k then it is in fact equivalent to some \kmDTA k m with computably many control locations (by~\cref{thm:k-DTA-char}).


%
	%

%


\section{Timed automorphisms and invariance}
\label{sec:inv}

A fundamental tool in this paper 
is invariance properties of timed languages recognised by \NTA with respect
to permutations of $\Rreal$ preserving integer differences.
In this section we establish these properties.
A \emph{timed automorphism} is a monotone bijection $\pi : \Rreal \to \Rreal$
\st for every $x \in \Rreal$, $\pi(x+1) = \pi(x)+1$. 
For instance, if $\pi(3.4)=4.5$, then necessarily $\pi(5.4) = 6.5$ and $\pi(-3.6) = -2.5$.
Timed automorphisms $\pi$ 
are extended point-wise to
timed words $\pi((a_1, t_1) \dots (a_n, t_n)) = (a_0, \pi(t_1)) \dots (a_n, \pi(t_n))$,
configurations $\pi(p, \mu, \now) = (p, \pi{\circ}\mu, \pi(\now))$,
transitions
$\pi(c \goesto {a, t} c') = \pi(c) \goesto {a, \pi(t)} \pi(c')$,
and sets $X$ thereof $\pi(X) = \setof{\pi(x)}{x\in X}$.
%
\begin{remark}
  A timed automorphism $\pi$ can in general take a nonnegative real $t\geq 0$ to a negative one.
  %
  Whenever we write $\pi(x)$,
  we always implicitly assume that $\pi$ is defined on $x$.
\end{remark}

Let $S \subseteq \R$.
An \emph{$S$-timed automorphism} is a timed automorphism \st $\pi(t) = t$ for all $t\in S$.
Let $\Pi_S$ denote the set of all $S$-timed automorphisms, and let $\Pi = \Pi_\emptyset$.
A set $X$ is \emph{$S$-invariant}
if $\pi(X) = X$ for every $\pi \in \Pi_S$;
equivalently,
for every $\pi \in \Pi_S$,
$x\in X$ if, and only if $\pi(x)\in X$.
%
A set $X$ is \emph{invariant} if it is $S$-invariant with $S = \emptyset$.
The following three facts express some basic invariance properties.
%
\begin{restatable}{fact}{invarianceTrans}
  \label{fact:equivariant:trans}
  The timed transition system $\semd A$ is invariant.
\end{restatable}
\noindent
By unrolling the definition of invariance in the previous fact,
we obtain that the set of configurations is invariant,
the set of transitions ${\goesto{}}$ is invariant,
and that the set of final configurations $F$ is invariant.
\begin{restatable}[Invariance of the language semantics]{fact}{factEquivLang}
    \label{fact:equivariant:lang}
    %
    The function $c \mapsto \langtsa A c$ from $\semd A$ to languages is invariant,
    i.e., for all timed permutations $\pi$,
    $\langtsa A {\pi(c)} = \pi(\langtsa A c)$.
\end{restatable}
%
%
\begin{restatable}[Invariance of the language of a configuration]{fact}{invarianceLang}
  \label{fact:invariantbase}
  \label{fact:invariantalways}
  The language $\langtsa A {p, \mu, \now}$ is
  $(\clockval c \cup\set{\now})$-invariant.  
  Moreover, if $A$ is always resetting,
  then $\langtsa A {p, \mu, \now}$ is
  $\clockval c$-invariant. 
\end{restatable}
%
%

Since timed automorphisms preserve integer differences,
only the fractional parts of elements of $S \subseteq \R$ matter
for $S$-invariance, and hence it makes sense to restrict to subsets of the half-open interval $[0, 1)$.
Let $\fract{S} = \setof{\fract{x}}{x\in S} \subseteq [0,1)$ stand for the set of fractional parts of elements of $S$.
The following lemma shows that,
modulo the irrelevant integer parts,
there is always the least set $S$ witnessing $S$-invariance.
\begin{restatable}{lemma}{lemLeastSup}
  \label{lem:leastsup}
  For finite subsets $S, S' \subseteq \R$, 
  if a timed language $L$ is both $S$-invariant and $S'$-invariant, 
  then it is also $S''$-invariant 
  where $S'' = \fract{S} \cap \fract{S'}$.
\end{restatable}

The \emph{$S$-orbit} of an element $x \in X$
(which can be an arbitrary object on which the action of
timed automorphisms is defined)
is the set 
$\orbit S x = \setof{\pi(x) \in X}{\pi \in \Pi_S}$
of all elements $\pi(x)$ which can be obtained by applying some $S$-automorphism to $x$.
The \emph{orbit} of $x$ is just its $S$-orbit with $S = \emptyset$,
written $\orbit {} x$.
Clearly $x$ and $x'$ have the same $S$-orbit
if, and only if, $\pi(x) = x'$ for some 
 $\pi \in \Pi_S$.
For greedily resetting \NTA,
orbits of single configurations are in bijective correspondence with bounded regions.
\begin{fact} \label{fact:reg}
  %
  Assume $A$ is a greedily resetting \kmNTA k m.
  Two reachable configurations $(p, \mu, \now)$ and $(p, \mu', \now')$ of $A$ with the same control location $p$
  have the same orbit
  if, and only if,
  the corresponding clock valuations $\assigntoval \mu {\now}$ and
  $\assigntoval {\mu'} {\now'}$
  belong to the same $k,m$-region.
\end{fact}

The \emph{$S$-closure} of a set $Y$, 
written $\closure S Y = \bigcup_{x \in Y} \orbit S x$,
is the union of the $S$-orbits of all its elements.
The following fact characterises invariance in term of closures.
\begin{fact} \label{fact:inv:closure}
  A set $Y$ is $S$-invariant if, and only if, $\closure S Y = Y$.
\end{fact}
%
%
\begin{proof}
  Only if direction follows by the definition of $S$-invariance. For the converse direction observe that
  $\closure S X = X$ implies $\pi(X)\subseteq X$ for every $\pi\in\Pi_S$.
  The opposite inclusion follows by closure of $S$-timed automorphisms under inverse:
  $\pi^{-1}(X) \subseteq X$, hence $X\subseteq \pi(X)$.
\end{proof}

\section{Decidability of \kDTA{k} and \kmDTA k m membership for \kNTA{1}}
\label{sec:upperbound}

In this section we prove \cref{thm:kDTA:memb} thus
establishing decidability of the \kDTA{k} and \kmDTA k m membership problems for \kNTA{1}.
%
%
Both results are shown using the following key characterisation of \kDTA k languages as a subclass of
\kNTA 1 languages.
In particular, this characterisation provides a small bound on the number of control locations of a \kDTA k equivalent to a given \kNTA 1 (if any exists).
%
%
\begin{lemma} \label{thm:k-DTA-char}
	Let $A$ be a \kmNTA{1} m with $n$ control locations,
	and let $k\in\N$.
	The following conditions are equivalent:
	\begin{enumerate}
		\item $\langts A = \langts B$ for some always resetting \kDTA{k} $B$.	 \label{p1}
		\item For every timed word $w$, there is $S\subseteq\R$ of size at most $k$ \st
		the last timestamp of $w$ is in $S$ and 
		$w^{-1} \langts A$ is $S$-invariant. \label{p2}
		\item $\langts A = \langts B$ for some always resetting \kmDTA{k} m $B$
			with at most $f(k,m,n) = \reg k m \cdot 2^{n (2km+1)}$ control locations
			($\reg k m$ stands for the number of $k,m$-regions). \label{p3}
	\end{enumerate}
\end{lemma}

\noindent
The proof of \Cref{thm:kDTA:memb} builds on \Cref{thm:k-DTA-char}
and on the following fact:
\begin{lemma} \label{lem:sep}
	The \kDTA k and \kmDTA k m membership problems are both decidable for \DTA languages.
\end{lemma}
\begin{proof}
	We reduce to a deterministic separability problem.
	Recall that a language $S$ \emph{separates} two languages $L, M$
	if $L \subseteq S$ and $S \cap M = \emptyset$.
	It has recently been shown that the $\kDTA k$ and $\kmDTA k m$ separability problems
	are decidable for \NTA \cite[Theorem 1.1]{ClementeLasotaPiorkowski:ICALP:2020},
	and thus, in particular, for \DTA.
	To solve the membership problem,
	given a \DTA $A$, the procedure computes a \DTA $A'$ recognising the complement of $\lang A$ and checks whether 
	$A$ and $A'$ are \kDTA k separable (resp.,~\kmDTA k m separable)
	by using the result above.
	It is a simple set-theoretic observation that
	$\langts A$ is a \kDTA k language if, and only if, the languages $\langts A$ and $\langts {A'}$ are separated
	by some \kDTA k language, and likewise for \kmDTA k m languages.
\end{proof}
\begin{proof}[Proof of \Cref{thm:kDTA:memb}]
	We solve both problems in essentially the same way.
	Given a \kmNTA 1 m $A$, 
	the decision procedure enumerates all always resetting \kmDTA{k+1} m $B$ with at most $f(k,m,n)$ 
	locations
	and checks whether $\langts A = \langts B$
	(which is decidable by \cite{OW04}).
	If no such \kDTA{k+1} $B$ is found, 
	the $\langts A$ is not an always resetting \kDTA {k+1} language, due to \Cref{thm:k-DTA-char},
	and hence forcedly is not a \kDTA k language either;
	the procedure therefore answers negatively.
	Otherwise, in case when such a \kDTA{k+1} $B$ is found, then \kDTA{k} membership
	(resp.~\kmDTA k m membership) test is performed on $B$, decidable due to~\Cref{lem:sep}.
\end{proof}

\begin{remark}[Complexity]
	The decision procedure for \kNTA{1} invokes the \HyperAckermann~subroutine of \cite{OW04} to check equivalence between a \kNTA 1 and a candidate \DTA.
	This is in a sense unavoidable,
    since we show in \Cref{lem:easy-undecidability} that the \kDTA k and \kmDTA k m membership problems are
    \HyperAckermann-hard for \kNTA{1}.
\end{remark}

In the rest of this section we present the proof of \Cref{thm:k-DTA-char}.
Let us fix a \kmNTA{1}{m} $A = \tuple{\Sigma, \L, \set{\x}, \I, \F, \Delta}$, where $m$ is the greatest constant used in clock constraints in $A$, and $k\in\N$. 
We assume w.l.o.g.~that $A$ is greedily resetting:
This is achieved by resetting the clock as soon as upon reading an input symbol
its value becomes greater than $m$ or is an integer $\leq m$;
we can record in the control location the actual integral value if it is $\leq m$, or a special flag otherwise.
Consequently, after every discrete transition the value of the clock is at most $m$,
and if it is an integer then it equals 0.

The implication \ref{p3}$\implies$\ref{p1} follows by definition.
For the implication \ref{p1}$\implies$\ref{p2} suppose, by assumption, 
$\langts A = \langts {B}$ for a total always resetting \kDTA{k} $B$.
Every left quotient $w^{-1} \langts A$ equals $\langtsa B c$ for some configuration $c$, hence
Point~\ref{p2} follows by \cref{fact:invariantalways}.
Here we use the fact that $B$ is always resetting
in order to apply the second part of \cref{fact:invariantalways};
without the assumption, we would only have $S$-invariance for sets $S$ of size at most $k+1$.

It thus remains to prove the implication \ref{p2}$\implies$\ref{p3},
which will be the content of the rest of the section.
Assuming Point~\ref{p2}, we
are going to define an always resetting \kmDTA k m $B'$ with clocks $\X = \set{\x_1, \dots, \x_{k}}$ and 
with at most $f(k,m,n)$ locations such that $\langts {B'} = \langts A$. 
We start from the timed transition system $\XX$ obtained by the finite powerset construction underlying
the determinisation of $A$, and then transform
this transition system gradually, while preserving its language, 
until it finally becomes isomorphic to the reachable part of $\semd {B'}$ for some \kmDTA k m $B'$.
As the last step we extract from this deterministic timed transition system a syntactic definition of $B'$ and prove equality
of their languages.
This is achievable due to the invariance properties witnessed by the
transition systems in the course of the transformation.

\para{Macro-configurations}
Configurations of the \kNTA 1 $A$ are of the form $c = (p, u, \now)$ where $u, \now \in \R$ and $u \leq \now$.
A \emph{macro-configuration} is a (not necessarily finite)
set $X$ of configurations $(p, u, \now)$ of $A$ which share the same value
of the current timestamp $\now$, which we denote as $\nowof X = \now$.
We use the notation 
$\langtsa A X := \bigcup_{c \in X} \langtsa A c$.
%
Let $\succe {a,t} X := \setof{c'\in\Conf A}{c \goesto{a,t} c' \text{ for some }c\in X}$
be the set of successors of configurations in $X$.
We define a deterministic timed transition system $\XX$ consisting of
the macro-configurations reachable in the course of determinisation of $A$.
Let $\XX$ be the smallest set of macro-configurations and transitions such that

\begin{itemize}
\item $\XX$ contains the initial macro-configuration: $X_0 = \setof{(p, 0, 0)}{p\in \I} \in \XX$;
\item $\XX$ is closed under successor: for every $X\in\XX$ and $(a,t)\in\Sigma\times\R$, there is a transition
$X \goesto {a,t} \succe {a,t} X$ in $\XX$.
\end{itemize}

\noindent
Due to the fact that $\semd A$ is finitely branching, i.e.~$\succe {a,t} {\set{c}}$ is finite for every fixed $(a,t)$,
all macro-configurations $X\in\XX$ are finite.
Let the final configurations of $\XX$ be $F_\XX = \setof{X\in\XX}{X\cap F \neq \emptyset}$.
\begin{claim} \label{claim:eqlangAX}
$\langtsa A X = \langtsa \XX {X}$ for every $X\in\XX$. In particular $\langts A = \langtsa \XX {X_0}$.
\end{claim}
For a macro-configuration $X$ 
we write 
$\macval {X} := \setof{u}{(p, u, \nowof X) \in X}\cup \set{\nowof X}$ to denote the reals appearing in $X$.
Since $A$ is greedily resetting, every macro-configuration
$X \in \XX$ satisfies
$\macval X \subseteq (\nowof X - m, \nowof X]$.
Whenever a macro-configuration $X$ satisfies this condition
we say that \emph{the span of $X$ is bounded by $m$}.

\para{Pre-states}
%
%
By assumption (Point 2),
$\langtsa A X$ is $S$-invariant for some $S$ of size at most $k$,
but the macro-configuration $X$ itself needs not be $S$-invariant in general.
Indeed, a finite macro-configuration $X\in\XX$ is $S$-invariant if, and only if, 
$\fract{\macval X} \subseteq \fract{S}$,
which is impossible in general when $X$ is arbitrarily large, its span is bounded (by $m$), and
size of $S$ is bounded (by $k$). 
Intuitively, in order to assure $S$-invariance we will replace $X$ by its $S$-closure $\closure S X$
(recall Fact~\ref{fact:inv:closure}).

A set $S\subseteq \R$ is \emph{fraction-independent}
if it contains no two reals with the same fractional part.
A \emph{pre-state} is a pair $Y = (X, S)$,
where $X$ 
is an $S$-invariant macro-state, 
and $S$ is a finite fraction-independent subset of $\macval X$ that contains $\nowof X$.
The intuitive rationale behind assuming the $S$-invariance of $X$ is that it implies, together with the bounded span of $X$
and bounded size of $S$,
that there are only finitely many pre-states,
up to timed automorphism. 
%
We define the deterministic timed transition system $\YY$ as the smallest set of pre-states
and transitions between them such that:

\begin{itemize}
\item $\YY$ contains the initial pre-state: $Y_0 = (X_0, \set{0}) \in \YY$; 
\item $\YY$ is closed under the closure of successor: for every $(X,S)\in\YY$ and $(a,t)\in\Sigma\times\R$,
there is a transition $(X, S) \goesto {a,t} (X', S')$, where
$S'$ is the least, with respect to set inclusion, subset of $S\cup\set{t}$ containing $t$ such that 
the language $L' = (a,t)^{-1} \langtsa A X = \langtsa A {\succe {a,t} X}$ is $S'$-invariant, and
$X' = \closure {S'} {\succe {a,t} X}$.
\end{itemize}

\begin{example}
\label{ex:B}
   Suppose $k=3$, $m=2$, $\succe {a,t} X = \set{(p,3.7,5),(q,3.9,5),(r,4.2,5)}$ and
   $S' = \set{3.7, 4.2, 5}$.
   Then $X' = \set{(p,3.7,5)}\cup\set{(q,t,5)\mid t \in (3.7, 4)}\cup\set{(r,4.2,5)}$.   
   $\nowof {X'} = 5$.
    A corresponding \emph{state} 
    is $(X', \mu')$, where $\mu' = \set{\x_1 \mapsto 3.7, \x_2 \mapsto 4.2, \x_3 \mapsto 5}$.
%
\end{example}

\noindent
Observe that the least such fraction-independent subset $S'$ exists due to the following facts: 
as $X$ is $S$-invariant, due to \cref{fact:equivariant:lang} so is its language
$\langtsa A X$, and hence $L'$ is necessarily $(S\cup\set{t})$-invariant;
by assumption (Point 2), $L'$ is $R$-invariant for some set $R\subseteq\R$ of size at most $k$ containing $t$;
let $T\subseteq \R$ be the least set given by Lemma~\ref{lem:leastsup}, i.e., 
$\fract{T}\subseteq \fract{S} \cap \fract{R}$; and finally let $S'\subseteq S$ be chosen so that
$\fract{S'} = \fract{T\cup\set{t}}$. 
Due to fraction-independence of $S$ the choice is unique, $S'$ is fraction-independent, and $t\in S'$.
Furthermore, the size of $S'$ is at most $k$.
By \cref{fact:equivariant:lang}, we deduce:
\begin{restatable}[Invariance of $\YY$]{claim}{claimEquivY}
\label{claim:equivY}
For every two transitions $(X_1, S_1) \goesto {a,t_1} (X'_1, S'_1)$ and $(X_2, S_2) \goesto {a,t_2} (X'_2, S'_2)$ in $\YY$
and a timed permutation $\pi$, if $\pi(X_1) = X_2$ and $\pi(S_1) = S_2$ and $\pi(t_1) = t_2$, then we have
$\pi(X'_1) = X'_2$ and $\pi(S'_1) = S'_2$.
\end{restatable}
Let the final configurations of $\YY$ be $F_\YY = \setof{(X,S)\in\YY}{X\cap \F \neq \emptyset}$.
By induction on the length of timed words it is easy to show:
\begin{claim} \label{claim:eqlangXY}
	$\langtsa \XX {X_0} = \langtsa \YY {Y_0}$.
\end{claim}
Due to the assumption that $A$ is greedily resetting and due to Point 2,
in every pre-state $(X, S) \in \YY$ the span of $X$ 
is bounded by $m$ and the size of $S$ is bounded by $k$.

\para{States}
We now introduce \emph{states}, which are designed to be in one-to-one correspondence 
with configurations of the forthcoming \kDTA k $B'$. 
Intuitively, a state differs from a pre-state $(X, S)$ only by allocating the values from $S$ into $k$ clocks,
thus while a pre-state contains a set $S$, the corresponding state contains a clock assignment $\mu : \X \to \R$ 
with image $\mu(\X) = S$.

Let $\X = \set{\x_1, \dots, \x_k}$ be a set of $k$ clocks.
A \emph{state} is a pair $Z = (X, \mu)$,
where $X$ 
is a macro-configuration, 
$\mu : \X \to \macval X$ is a clock reset-point assignment,
$\mu(\X)$ is a fraction-independent set containing $\nowof X$,
and $X$ is $\mu(\X)$-invariant.
Thus every state $Z = (X, \mu)$ determines uniquely a corresponding pre-state 
$\sigma(Z) = (X, S)$ with $S=\mu(\X)$.
We define the deterministic timed transition system $\ZZ$
consisting of those states $Z$ \st $\sigma(Z) \in \YY$, and of transitions determined as follows:
$(X, \mu) \goesto {a,t} (X', \mu')$ if the corresponding pre-state has a transition 
$(X, S) \goesto {a,t} (X', S')$ in $\YY$, where $S = \mu(\X)$, and
\begin{align} \label{eq:mu}
\mu'(\x_i) \ := \ \begin{cases}
t &\text{ if } \mu(\x_i) \notin S' \text{ or } \mu(\x_i) = \mu(\x_j) \text{ for some } j>i\\
\mu(\x_i) & \text{otherwise.}
\end{cases}
\end{align}
Intuitively, the equation~\eqref{eq:mu} defines a deterministic update of the clock reset-point assignment $\mu$ 
that amounts to resetting ($\mu'(\x_i):=t$) all clocks
$\x_i$ whose value is either no longer needed (because $\mu(\x_i) \notin S'$), or
is shared with some other clock $x_j$, for $j>i$ and is thus redundant.
Due to this disciplined elimination of redundancy, 
knowing that $t\in S'$ and the size of $S'$ is at most $k$, 
we ensure that at least one clock is reset in every step. In consequence, $\mu'(\X) = S'$,
and the forthcoming \kDTA k $B'$ will be always resetting.
%
Using Claim~\ref{claim:equivY} 
we derive:
\begin{restatable}[Invariance of $\ZZ$]{claim}{claimEquivZ}
\label{claim:equivZ}
For every two transitions $(X_1, \mu_1) \goesto {a,t_1} (X'_1, \mu'_1)$ and 
$(X_2, \mu_2) \goesto {a,t_2} (X'_2, \mu'_2)$ in $\ZZ$
and a timed permutation $\pi$, if $\pi(X_1) = X_2$ and $\pi{\circ}{\mu_1} = \mu_2$ and $\pi(t_1) = t_2$, then we have
$\pi(X'_1) = X'_2$ and $\pi{\circ}{\mu'_1} = \mu'_2$.
\end{restatable}
Let the initial state be $Z_0 = (X_0, \mu_0)$, where $\mu_0(\x_i) = 0$ for all $\x_i \in \X$, and
let final states be $F_\ZZ = \setof{(X, \mu)\in \ZZ}{X \cap F \neq \emptyset}$. 
By induction on the length of timed words one proves:
\begin{claim} \label{claim:eqlangYZ}
$\langtsa \YY {Y_0} = \langtsa \ZZ {Z_0}$.
\end{claim}
%
%
In the sequel we restrict $\ZZ$ to states reachable from $Z_0$.
In every state $Z  = (X, \mu)$ in $\ZZ$,
we have $\nowof X \in \mu(\X)$.
This will ensure the resulting \kDTA k $B'$ to be always resetting.

\para{Orbits of states}

While a state is designed to correspond to a configuration of the forthcoming \kDTA k $B'$,
its orbit is designed to play the role of control location of $B'$. 
We therefore need to prove that the set of states in $\ZZ$ is orbit-finite, i.e., 
the set of orbits 
$\setof{\orbit {} Z}{Z\in\ZZ}$
is finite and its size is bounded by $f(k,m,n)$.
We start by deducing an analogue of Fact~\ref{fact:reg}:
\begin{claim} \label{claim:reg}
For two states $Z = (X, \mu)$ and $Z' = (X', \mu')$ in $\ZZ$, their clock assignments are in the same orbit,
i.e., $\pi{\circ}\mu = \mu'$ for some $\pi\in\Pi$, if, and only if,
the corresponding  clock valuations $\assigntoval \mu {\nowof X}$ and
$\assigntoval {\mu'} {\nowof {X'}}$
belong to the same $k,m$-region.
\end{claim}
(In passing note that, since in every state $(X, \mu)$ in $\ZZ$ the span of $X$ is bounded by $m$, 
only bounded $k,m$-regions can appear in the last claim. 
Moreover, in each of $k,m$-regions one of clocks equals $0$.)
The action of timed automorphisms on macro-configurations and clock assignments is extended to states as $\pi(X, \mu) = (\pi(X), \pi{\circ}\mu)$.
Recall that the orbit of a state $Z$ is defined as $\orbit {}  {Z} = \setof{\pi(Z)}{\pi \in \Pi}$.
\begin{claim}
The number of orbits of states in $\ZZ$ is bounded by $f(k,m,n)$.
\end{claim}
\begin{proof}
	We finitely represent a state $Z = (X, \mu)$, relying on the following general fact.
	\begin{fact}
		For every $u \in\R$ and $S\subseteq \R$, 
		the $S$-orbit\footnote{The orbits of states $Z$ should not be confused with 	$S$-orbits of individual reals $u\in\R$.} 
		$\orbit S u$ 
		is either the singleton $\set{u}$ (when $u\in S$) or an open interval
		with ends-points of the form $t + z$ where $t \in S$ and $z\in\Z$ (when $u\notin S$).
	\end{fact}
	We apply the fact above to $S = \mu (\X)$.
	In our case the span of $X$ is bounded by $m$,
	and thus the same holds for $\mu (\X)$.
	Consequently, the integer $z$ in the fact above
	always belongs to $\set{-m, -m{+}1, \dots, m}$.
	In turn, $X$ splits into disjoint $\clockval X$-orbits 
	$\orbit {\clockval X} u$ consisting of open intervals 
	separated by endpoints of the form $t + z$
	where $t \in \mu(\X)$ and $z\in\set{-m, -m{+}1, \dots, m}$.

\begin{example}
	Continuing Example~\ref{ex:B}, the endpoints are $\set{3, 3.2, 3.7, 4, 4.2, 4.7, 5}$, as shown in the illustration:
	\begin{center}
		\includegraphics[scale=0.5]{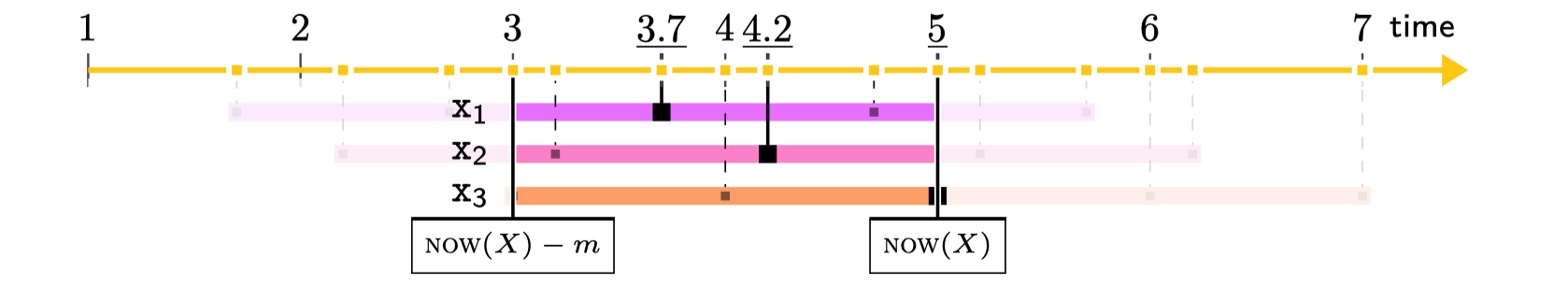}
	\end{center}
\end{example}

	\noindent
	Recall that $\mu(\X)$ is fraction-independent.
	Let $e_1 < e_2 < \dots < e_{l+1}$ be all the endpoints of open-interval orbits ($l \leq km$), and let 
	$o_1, o_2, o_3, \dots \ := \ \set{e_1}, (e_1, e_2), \set{e_2}, \dots$ be the consecutive $S$-orbits
$\orbit {\mu(\X)} u$ of elements $u\in\clockval \X$. 
The number thereof is $2l+1 \leq 2km+1$.
The finite representation of $Z = (X, \mu)$ consists of 
the pair $\pair{O}{\mu}$, where
\begin{align} \label{eq:O}
	O = \set{(o_1, P_1), \dots, (o_{2l+1}, P_{2l+1})}
\end{align}
assigns to each orbit $o_i$ the set of locations 
$
P_i  = \setof{p}{(p, u, \now)\in X \text{ for some } u\in o_i} \subseteq \L,
$
(which is the same as
$
P_i = \setof{p}{(p, u, \now)\in X \text{ for all } u\in o_i}
$
since $X$ is $\mu(\X)$-invariant, and hence $\mu(\X)$-closed).
Thus a state $Z = (X, \mu)$ is uniquely determined by the sequence $O$ as in~\eqref{eq:O}
and the clock assignment $\mu$.

We claim that the set of all the finite representations $(O, \mu)$, as defined above, is orbit-finite.
Indeed, the orbit of $(O, \mu)$ is determined by the orbit of $\mu$ and the sequence
\begin{align} \label{eq:P}
P_1, \ P_2, \ \ldots, \ P_{2km+1}
\end{align}
induced by the assignment $O$ as in~\eqref{eq:O}.
Therefore, the number of orbits is bounded by the number of orbits of $\mu$
(which is bounded, due to Claim~\ref{claim:reg}, by $\reg k m$)
times the number of different sequences of the form~\eqref{eq:P} 
(which is bounded by $(2^n)^{2km+1}$).
This yields the required bound $f(k, m, n) = \reg k m \cdot 2^{n(2km +1)}$.
\end{proof}

\para{Construction of the \DTA}  
As the last step we define a \kDTA {k} $B' = \tuple{\Sigma, \L', \X, \set{o_0}, \F', \Delta'}$ 
such that the reachable part of $\semd {B'}$ is isomorphic to $\ZZ$.
Let locations $\L' = \setof{\orbit {} Z}{Z\in\ZZ}$ be orbits of states from $\ZZ$, the initial location be the orbit $o_0$ of $Z_0$, and final locations 
$\F' = \setof{\orbit{}{Z}}{Z\in F_\ZZ}$ be orbits of final states.
A transition $Z = (X, \mu) \goesto{a, t} (X', \mu') = Z'$ in $\ZZ$
induces a transition rule in $B'$
\begin{align} \label{eq:trofB}
	\transition{o}{a}{\psi}{\Y}{o'} \in \Delta'
\end{align}
whenever $o = \orbit{}{Z}$, $o' = \orbit{}{Z'}$,
$\psi$ is the unique $k,m$-region satisfying $\assigntoval \mu t \in \semlog{\psi}$,
and $\Y = \setof{\x_i\in\X}{\mu'(\x_i) = t}$.
The automaton $B'$ is indeed a \DTA since $o$, $a$ and $\psi$ uniquely determine $\Y$ and $o'$:
\begin{claim}
Suppose that two transitions 
$(X_1, \mu_1) \goesto {a,t_1} (X'_1, \mu'_1)$ and $(X_2, \mu_2) \goesto {a,t_2} (X'_2, \mu'_2)$ in $\ZZ$
%
induce transition rules $\transition{o}{a}{\psi}{\Y_1}{o'_1}, \transition{o}{a}{\psi}{\Y_2}{o'_2} \in \Delta'$
with the same source location $o$ and constraint $\psi$, i.e, 
\begin{align} \label{eq:psi}
\assigntoval {\mu_1} {t_1} \in\semlog{\psi} \qquad
\assigntoval {\mu_2} {t_2} \in\semlog{\psi}.
\end{align}
Then the target locations are equal $o'_1 = o'_2$, and the same for the reset sets $\Y_1 = \Y_2$.
\end{claim}
\begin{proof}
We use the invariance of semantics of $A$ and Claim~\ref{claim:equivZ}.
Let $o = \orbit{}{X_1, \mu_1} = \orbit{}{X_2, \mu_2}$.
Thus there is a timed automorphism $\pi$ such that
\begin{align} \label{eq:Y}
X_2 = \pi(X_1) \qquad 
\mu_2 = \pi{\circ}\mu_1.
\end{align}
%
%
It suffices to show that there is a (possibly different) timed permutation $\sigma$ satisfying the following equalities:
\begin{align} \label{eq:4toprove}
t_2 = \sigma(t_1) \quad
\setof{i}{\mu'_1(\x_i) = t_1} = \setof{i}{\mu'_2(\x_i) = t_2} \quad
\mu'_2 = \sigma{\circ}\mu'_1 \quad
X'_2 = \sigma(X'_1).
\end{align}
We now rely the fact that both ${\now}_1 = \nowof {X_1} \in \mu_1(\X)$ and ${\now}_2 = \nowof {X_2} \in \mu_2(\X)$
are assigned to (the same) clock due to the second equality in~\eqref{eq:Y}:
${\now}_1 = \mu_1(\x_i)$ and ${\now}_2 = \mu_2(\x_i)$.
We focus on the case when $t_1 - {\now}_1 \leq m$ (the other case is similar but easier as all clock are reset
due to greedy resetting),
which implies $t_2 - {\now}_2 \leq m$ due to~\eqref{eq:psi}.
In this case we may assume w.l.o.g., due to~\eqref{eq:psi} and the equalities~\eqref{eq:Y}, 
that $\pi$ is chosen so that $\pi(t_1) = t_2$. 
We thus take $\sigma = \pi$ for proving the equalities~\eqref{eq:4toprove}.
Being done with the first equality, we observe that the last two equalities in~\eqref{eq:4toprove} hold
due to the invariance of $\ZZ$ (cf.~Claim~\ref{claim:equivZ}).
The remaining second equality in~\eqref{eq:4toprove} is a consequence of the third one.
\end{proof}
%
%
\begin{claim}
	\label{claim:last}
	Let $Z = (X, \mu)$ and $Z' = (X', \mu)$ be two states in $\ZZ$ with the same clock assignment.
	If $\pi(X) = X'$ and $\pi{\circ}\mu = \mu$ for some timed automorphism $\pi$ then $X = X'$.
\end{claim}
\begin{claim} \label{claim:eqlangCB}
	$\ZZ$ is isomorphic to the reachable part of $\semd {B'}$.
\end{claim}
\begin{proof}
	For a state $Z = (X, \mu)$,
	let $\confof Z = (o, \mu, t)$, where $o = \orbit{}{Z}$ and $t = \nowof X$.
	By \cref{claim:last},
	the mapping $\confof \_$ is a bijection between $\ZZ$ and its image $\confof{\ZZ} \subseteq \semd {B'}$.
	By \eqref{eq:trofB},
	$\ZZ$ is isomorphic to a subsystem of the reachable part of $\semd {B'}$.
	The converse inclusion follows by the observation that $\ZZ$ is total: for every 
	$(a_1, t_1) \dots (a_n, t_n) \in\timedwords \Sigma$,
	there is a sequence of transitions $(X_0, \mu_0) \goesto {a_1, t_1} \cdots \goesto{a_n, t_n}$ in $\ZZ$.
\end{proof}
Claims~\ref{claim:eqlangAX}, \ref{claim:eqlangXY}, \ref{claim:eqlangYZ},
and~\ref{claim:eqlangCB} prove $\langts A = \langts {B'}$.
%

\section{Undecidability and hardness}
\label{sec:lowerbound}

In this section we complete the decidability status of the deterministic membership problem
by providing matching undecidability and hardness results.
In \cref{sec:undecidability} we prove undecidability of the \kDTA membership problem for \kNTA 1 (c.f.~\cref{thm:undecidability})
and in \cref{sec:hardness} we prove \HyperAckermann-hardness of the \kDTA{k} membership problem for \kNTA{1} (c.f.~\cref{thm:hardness}).


\subsection{Undecidability of \DTA and \mDTA m membership for \kNTA{1}} 
\label{sec:undecidability}

It has been shown in \cite[Theorem 1]{Finkel:FORMATS:2006} that it is undecidable
whether a \kNTA k timed language can be recognised by some \DTA, for any fixed $k \geq 2$.
This was obtained by a reduction from the \kNTA k universality problem,
which is undecidable for any fixed $k \geq 2$.
While the universality problem becomes decidable for $k = 1$,
we show in this section that, as announced in \cref{thm:undecidability},
the \DTA membership problem remains undecidable for \kNTA 1.



Since the universality problem for \kNTA 1 is decidable,
we need to reduce from another (undecidable) problem.
Our candidate is the finiteness problem of lossy counter machines,
which is undecidable \cite[Theorem 13]{Mayr:TCS:2003}.
A \emph{$k$-counters lossy counter machine} (\kLCM k) is a tuple $M = \tuple {C, Q, q_0, \Delta}$,
where $C = \set{c_1, \dots, c_k}$ is a set of $k$ counters,
$Q$ is a finite set of control locations,
$q_0 \in Q$ is the initial control location,
and $\Delta$ is a finite set of instructions of the form $\tuple {p, \op, q}$,
where $\op$ is one of $\incr c$, $\decr c$, and $\ztest c$.
A configuration of an \LCM $M$ is a pair $\tuple {p, u}$,
where $p \in Q$ is a control location,
and $u \in \N^C$ is a counter valuation.
For two counter valuations $u, v \in \N^C$,
we write $u \leq v$ if $u(c) \leq v(c)$ for every counter $c \in C$.
The semantics of an \LCM $M$ is given by a (potentially infinite) transition system over the configurations of $M$
\st there is a transition $\tuple {p, u} \goesto \delta \tuple {q, v}$,
for
$\delta = \tuple {p, \op, q} \in \Delta$, whenever
\begin{inparaenum}[1)]
    \item $\op = \incr c$ and $v \leq \extend c {u(c) + 1} u$, or
    \item $\op = \decr c$ and $v \leq \extend c {u(c) - 1} u$, or
    \item $\op = \ztest c$ and $u(c) = 0$ and $v \leq u$.
\end{inparaenum}
%
The \emph{finiteness problem} (a.k.a.~space boundedness) for an \LCM $M$
asks to decide whether the reachability set
$\reachset M = \setof{\tuple{p, u}} {\tuple {q_0, u_0} \goesto {}^* \tuple{p, u}}$
is finite, where $u_0$ is the constantly $0$ counter valuation.
\begin{theorem}[\protect{\cite[Theorem 13]{Mayr:TCS:2003}}]
    The \kLCM 4 finiteness problem is undecidable.
\end{theorem}

\noindent
We use the following encoding of \LCM runs
as timed words over the alphabet $\Sigma = Q \cup \Delta \cup C$
(c.f.~\cite[Definition 4.6]{LasotaWalukiewicz:ATA:ACM08} for a similar encoding).
We interpret a counter valuation $u \in \N^C$
as the word over $\Sigma$
\begin{align*}
    u \;=\;
        \underbrace {c_1 c_1 \cdots c_1}_{u(c_1) \text{ letters}}\ 
            \underbrace {c_2 c_2 \cdots c_2}_{u(c_2) \text{ letters}}\ 
				\underbrace {c_3 c_3 \cdots c_3}_{u(c_3) \text{ letters}}\ 
					\underbrace {c_4 c_4 \cdots c_4}_{u(c_4) \text{ letters}}.
\end{align*}
With this interpretation, we encode an \LCM run
%
    $\pi \;=\; \tuple{p_0, u_0}
        \goesto {\delta_1} \tuple{p_1, u_1}
            \goesto {\delta_2} \cdots
                \goesto {\delta_n} \tuple{p_n, u_n}$
%
as the following timed word,
called the \emph{reversal-encoding} of $\pi$,
\begin{align*}
    p_n \delta_n u_n\quad \cdots\quad p_1 \delta_1 u_1\quad p_0 u_0,
\end{align*}
\st $p_n$ occurs at time 0,
for every $1 \leq i < n$, $p_i$ occurs exactly after one time unit since $p_{i+1}$,
and if a ``unit'' of counter $c_1$ did not disappear due to lossiness when going from $u_i$ to $u_{i+1}$,
then the timestamps of the corresponding occurrences of letter $c_1$ in $u_i$ and $u_{i+1}$ are also at distance one
(and similarly for the other counters).
Under the encoding above, we can build a \kNTA 1 $A$ recognising the complement of the 
set of reversal-encodings of the runs of $M$
(\cite{LasotaWalukiewicz:ATA:ACM08} for more details about the construction of $A$).
Intuitively, when reading the reversal-encoding of a run of $M$, 
the counters are allowed to spontaneously increase.
Therefore, the only kind of error that $A$ must verify is that some counter spontaneously decreases.
This can be done by guessing an occurrence of letter (say) $c_1$ in the current configuration
which does not have a corresponding occurrence in the next configuration after exactly one time unit.
This check can be performed by an \NTA with one clock.
%
%
%
\begin{restatable}{lemma}{lemReduction}
	The set of reachable configurations $\reachset M$ is finite
	if, and only if,
	$\lang A$ is a deterministic timed language.
\end{restatable}

Since the timed automaton constructed in the proof uses only constant 1,
the reduction works also for the
\mDTA m membership problem for every $m> 0$:
\begin{corollary}
    For every fixed $m > 0$, the \mDTA m membership problem for \kNTA 1 languages is undecidable.
\end{corollary}
This result is the best possible in terms of the parameter $m$
since the problem becomes decidable for $m=0$.
In fact, the class of \kmDTA k 0 languages coincides with the class of \kmDTA 1 0 languages
(one clock is sufficient; c.f.~\cite[Lemma 19]{OW04}),
and thus \mDTA 0 membership reduces to \kmDTA 1 0 membership,
which is decidable for \kNTA 1 by \cref{thm:kDTA:memb}.

\begin{remark}
	\label{remark:compression}
	We observe that the reduction above uses a large alphabet $\Sigma$
	whose size depends on the input \LCM $M$.
	In fact, an alternative encoding exists using a unary alphabet $\Sigma = \set{a}$.
	Let the input \LCM $M$ have control locations $Q = \set{p_1, \dots, p_m}$
	and instructions $\Delta = \set{\delta_1, \dots, \delta_n}$.
	An \LCM configuration $p_{j} \delta_{k} u$ is represented by the timed word consisting of 6 blocks
	%
	$	\underbrace {a \cdots a}_{j \text{ letters}}
			\ \underbrace {a \cdots a}_{k \text{ letters}}
				\ \underbrace {a \cdots a}_{u(c_1) \text{ letters}}
					\ \underbrace {a \cdots a}_{u(c_2) \text{ letters}}
						\ \underbrace {a \cdots a}_{u(c_3) \text{ letters}}
							\ \underbrace {a \cdots a}_{u(c_4) \text{ letters}}$
	%
	\st in each block the last $a$ is at timed distance exactly one from 
	 the last $a$ of the previous block.
	A unit of counter $c_1$ now repeats at distance $6$ in the next configuration (instead of $1$).
	This shows that the \DTA membership problem is undecidable for
	\kNTA 1 using maximal constant $m=6$ over a unary alphabet.
\end{remark}

\subsection{Undecidability and hardness for \kDTA k and \kmDTA k m membership}
\label{sec:hardness}

All the lower bounds in this section
are obtained by a reduction from the universality problem for the respective language 
classes
(does a given language $L \subseteq \timedwords{\Sigma}$ satisfy $L = \timedwords{\Sigma}$?).
The reduction is a suitable adaptation, generalization, and simplification
of \cite[Theorem 1]{Finkel:FORMATS:2006}
showing undecidability of \DTA membership for \NTA languages.

A timed language $L$ is \emph{timeless} if $L = \langts A$ for $A\in$ \kNTA{0} a timed automaton with no clocks
(hence timestamps appearing in input words are irrelevant for acceptance).
For two languages $L \subseteq \timedwords \Sigma$ and $M \subseteq \timedwords \Gamma$, and a
fresh alphabet symbol $\$ \not\in \Sigma \cup \Gamma$, we define their \emph{composition}  $L \rhd \set{\$} \rhd M$
to be the following timed language over $\Sigma' = \Sigma \cup\set{\$} \cup\Gamma$:
        \[
        L \rhd \set{\$} \rhd M \ = \ \setof{v (\$, t) (a_1, t_1 + t) \dots (a_n, t_n + t)\in \timedwords{\Sigma'}}
        {v \in L, (a_1, t_1)\dots(a_n,t_n)\in M}.
        \]

%

\begin{lemma}[restate = lemEasyUndec, name = ]\label{lem:easy-undecidability}
	Let $k ,m \in \N$ and let $\mathcal Y$ be a class of timed languages that
	\vspace{-\topsep} \begin{enumerate}
	\item contains all the timeless timed languages,
	\item is closed under union and composition, and 
	\item contains some non-\kDTA{k} (resp.~non-\kmDTA k m) language.
	\end{enumerate}
	The universality problem for languages in $\mathcal Y$ reduces in polynomial time
	to the \kDTA{k} (resp. \kmDTA k m) membership problem for languages in $\mathcal Y$.
\end{lemma}
%
%
%
We immediately obtain \cref{thm:hardness} as a corollary of \cref{lem:easy-undecidability},
thanks to the following observations.
First, the lemma is applicable by taking as $\mathcal Y$
the classes of languages recognised by \kNTA 1
since this class contains all timeless timed languages,
is closed under union and composition, and
is not included in \kDTA{k} for any $k$ nor in \kmDTA k m for any $k,m$
(c.f.~the \kNTA{1} language from \cref{example:L1}
which is not recognised by any \DTA).
Second, \HyperAckermann-hardness of the universality problem for \kNTA{1} follows form the
same lower bound for the reachability problem in lossy channel systems~\cite[Theorem 5.5]{CS08},
together with the reduction from this problem to universality of \kNTA{1} given in~\cite[Theorem 4.1]{LasotaWalukiewicz:ATA:ACM08}.

Since the universality problem is undecidable for \kNTA{2} \cite[Theorem 5.2]{AD94}
and \kNTAe{1} (\kNTA 1 with epsilon steps) \cite[Theorem 5.3]{LasotaWalukiewicz:ATA:ACM08},
using the same reasoning we can apply \cref{lem:easy-undecidability}
to observe that the \kDTA k and \kmDTA k m membership problems are undecidable for \kNTA{2} and \kNTAe{1},
which refines the analysis of \cite[Theorem 1]{Finkel:FORMATS:2006}.
%


\section{Conclusions}

We have shown decidability and undecidability results
for several variants of the deterministic membership problem for timed automata.
Regarding undecidability, we have extended the previously known results
\cite{Finkel:FORMATS:2006,Tripakis:IPL:2006}
by proving that the \DTA membership problem is undecidable already for \kNTA 1 (\cref{thm:undecidability}),
and, over a unary input alphabet, it is undecidable for \kmNTA 1 m with $m \geq 6$ (\cref{remark:compression}).
We leave open the question of what is the minimal $m$ guaranteeing undecidability.
Regarding decidability, we have shown that when the resources available to the deterministic automaton are fixed
(either just the number of clocks $k$, or both clocks $k$ and maximal constant $m$),
then the respective deterministic membership problem is decidable (\cref{thm:kDTA:memb})
and \HyperAckermann-hard (\cref{thm:hardness}).

Our deterministic membership algorithm is based on a characterisation of \kNTA 1 languages
which happen to be \kDTA k (\cref{thm:k-DTA-char}),
which is proved using a semantic approach leveraging on notions from the theory of sets with atoms \cite{BL12}.
Analogous decidability results for register automata can be obtained with similar techniques.
It would be interesting to compare this approach
to the syntactic determinisation method of \cite{BBBB:ICALP:2009}.

Finally, our decidability results extend to the slightly more expressive class of always resetting \kNTA 2,
which have intermediate expressive power strictly between \kNTA 1 and \kNTA 2.

\newpage

\bibliography{bib}

\appendix


\section{Proofs for \cref{sec:inv}}

\invarianceTrans*
\begin{proof}
  Suppose $c = (p, \mu, \now) \goesto {a,t} (p', \mu', t) = c'$ due to some transition rule of $A$
  whose clock constraint $\varphi$ compares values of clocks
  $\x$, i.e., the differences $t - \mu(\x)$, to integers.
  Since a timed automorphism $\pi$ preserves integer distances, the same clock constraint is satisfied
  in $\pi(c) = (p, \pi{\circ}{\mu}, \pi(\now))$, and therefore the same transition rule is applicable yielding the transition
  $(p, \pi{\circ}{\mu}, \pi(\now)) \goesto {a,\pi(t)} (p, \pi{\circ}{\mu'}, \pi(t)) = \pi(c')$.
\end{proof}

\invarianceLang*
\begin{proof}
  This is a direct consequence of the invariance of semantics.
  Indeed, for every $(\clockval c\cup\set{\now})$-timed permutation $\pi$ the configurations $c = (p, \mu, \now)$ and
  $\pi(c) = (p, \pi{\circ}\mu, \pi(\now))$ are equal,
  hence their languages $\langtsa A {c}$ and
  $\langtsa A {\pi(c)}$, the latter equal to $\pi(\langtsa A c)$ by Fact~\ref{fact:equivariant:lang}, are equal too.
  Thus, $L = \pi(L)$.
  Finally, if $A$ is always resetting, then $\now \in\clockval c$,
  from which the second claim follows.
\end{proof}

\factEquivLang*
\begin{proof}
Consider a timed permutation $\pi$ and an accepting run of $A$ over a timed word 
$w = (a_1, t_1)  \dots (a_n, t_n) \in \timedwordsafter \Sigma \now$ starting in $c = (p, \mu, \now)$:
\begin{align*} 
           (p, \mu, \now) 
     \goesto {a_1, t_1} \,
     \cdots
                    \goesto{a_n, t_n} \,
                        (q, \nu, t_n),
\end{align*}
After $a_i$ is read, the value of each clock is either the difference $t_i - \mu(\x)$ for some $1\leq i \leq n$ and clock $\x\in\X$, 
or the difference $t_i - t_j$ for some $1 \leq j \leq i$.
Likewise is the difference of values of any two clocks.
Thus clock constraints of transition rules used in the run compare these differences to integers.
As timed automorphism $\pi$ preserves integer differences, by executing the same sequence of transition rules we obtain
the run over $\pi(w)$ starting in $\pi(c) = (p, \pi{\circ}\mu, \pi(\now))$: 
\[
           (p, \pi{\circ}\mu, \pi(\now)) 
   \goesto {a_1, \pi(t_1)} \,
   \cdots
                    \goesto{a_n, \pi(t_n)} \,
                        (q, \pi{\circ}\nu, \pi(t_n)),
\]
also accepting as it ends in the same location $q$.
As $w\in\timedwords \Sigma$ can be chosen arbitrarily, we have thus proved one of inclusions, namely 
\[
\pi(\langtsa A {p, \mu, \now}) \ \subseteq \ \langtsa A {p, \pi{\circ}\mu, \pi(\now)}.
\]
The other inclusion follows from the latter one applied to $\pi^{-1}$ and $\langtsa A {p, \pi{\circ}\mu, \pi(\now)}$:
\[
\pi^{-1}(\langtsa A {p, \pi{\circ}\mu, \pi(\now)}) \ \subseteq  \ \langtsa A {p, \pi^{-1}{\circ}\pi{\circ}\mu, \pi^{-1}(\pi(\now))}
\ = \ \langtsa A {p, \mu, \now}.
\]
The two implications prove the equality.
\end{proof}

\lemLeastSup*
\begin{proof}
Let $L$ be an $S$- and $S'$-invariant timed language, and let 
%
$F = \fract{S}$ and $F' = \fract{S'}$.
Towards proving that $L$ is an $(F \cap F')$-invariant subset of $\timedwords \Sigma$, 
consider two timed words $w, w' \in \timedwords \Sigma$ such that $w' = \pi(w)$ for some $(F\cap F')$-timed 
automorphism $\pi$.
We need to show that $w \in L$ iff $w' \in L$, which follows immediately by the following claim:
\begin{claim}
Every $(F\cap F')$-timed automorphism $\pi$ decomposes into $\pi = \pi_n \circ \dots \circ \pi_1$, where
each $\pi_i$ is either $F$- or $F'$-timed automorphism.
\end{claim}
Indeed, due to $F$- and $F'$-invariance of $L$, we have $w \in L$ iff $w' \in L$ as required.

As it has been proved in~\cite{BKL12fulllics}, instead of dealing with decomposition of $\pi$,
it is sufficient to analyse the individual orbit of $F - F'$, in the special case when 
both $F - F'$ and $F' - F$ are singleton sets.
The proof of Theorem 10.3 in~\cite{BKL12fulllics} may be repeated here to prove 
that the last claim above is implied by the following one:
\begin{claim}
Let $F, F' \subseteq [0, 1)$ be finite sets \st $F - F' = \set{t}$ and $F' - F = \set{t'}$.
For every $(F\cap F')$-timed automorphism $\pi$ we have $\pi(t) = (\pi_n \circ \dots \circ \pi_1)(t)$, 
for some $\pi_1, \dots, \pi_n$, each of which is either $F$- or $F'$-timed automorphism.
\end{claim}
The proof of the claim is split into two cases.

\para{Case $F\cap F' \neq \emptyset$}
Let $l$ be the greatest element of $F\cap F'$ smaller than $t$, and let $h$ be the smallest element of $F\cap F'$ 
greater than $t$, assuming they both exist. 
(If $l$ does not exist put $l := h'-1$, where $h'$ is the greatest element of $F\cap F'$; 
symmetrically, if $h$ does not exists put $h := l'+1$, where $l'$ is the smallest element of $F \cap F'$.)
Then the $(F\cap F')$-orbit $\setof{\pi(t)}{\pi \text{ is a } (F\cap F')\text{-timed automorphism}}$
is the open interval  $(l,h)$.
Take any $(\F\cap F')$-timed automorphism $\pi$; without loss of generality assume that $u=\pi(t)>t$. 
The only interesting case is $t < t' \leq u$.  
In this case, we show $\pi_2(\pi_1(t))$,where
\begin{itemize}
\item $\pi_1$ is some $F'$-timed automorphism that acts as identity on $[t',l+1]$ and \st $t < \pi_1(t) < t'$,
\item $\pi_2$ is some $F$-timed automorphism that acts as identity on $[h-1, t]$ and \st $\pi_2(\pi_1(t))=u$.
\end{itemize}

\para{Case $F\cap F' = \emptyset$}
Thus $F = \set{t}$ and $F' = \set{t'}$.
Take any timed automorphism $\pi$; without loss of generality assume that $\pi(t)>t$.
Let $z\in\Z$ be the unique integer \st $t' + z -1 < t < t' + z$.
%
Let $\pi_1$ be an arbitrary $\set{t'}$-timed automorphism that maps $t$ to some $t_1 \in (t,t'+z)$.
Note that $t_1$ may be any value in $(t, t'+z)$. 
Similarly, let $\pi_2$ be an arbitrary $\set{t}$-timed automorphism that maps $t_1$ to some $t_2 \in (t', t+1)$.
Again, $t_2$ may be any value in $(t', t+1)$.
By repeating this process sufficiently many times one finally reaches $\pi(t)$ as required.
\end{proof}

\section{Proofs for \cref{sec:upperbound}}

\claimEquivY*
\begin{proof} 
Let $i$ range over $\set{1, 2}$ and let  $\widetilde X_i := \succe {a,t_i} {X_i}$. Thus $S'_i$ is  the least subset of 
$S_i\cup\set{t_i}$ containing $t_i$ such that $\langtsa A {\widetilde X_i}$ is $S'_i$-invariant,
and $X'_i = \closure {S'_i} {\widetilde X_i}$.
By invariance of $\semd A$ (\cref{fact:equivariant:trans}) and invariance of semantics
(Fact~\ref{fact:equivariant:lang}) we get
\[
\pi(\widetilde X_1) = \widetilde X_2,
\qquad \text{ and } \qquad \pi(\langtsa A {\widetilde X_1}) = \langtsa A {\widetilde X_2},
\]
and therefore $\pi(S'_1) = S'_2$, which implies $\pi(X'_1) = X'_2$.
\end{proof}

\claimEquivZ*
\begin{proof}
Let $i$ range over $\set{1, 2}$. Let
$S_i = \mu_i(\X)$ and $(X_i, S_i) \goesto {a,t_i} (X'_i, S'_i)$ in $\YY$. By Claim~\ref{claim:equivY} we have
\[
\pi(X'_1) = X'_2 \qquad \text{ and } \pi(S'_1) = S'_2.
\]
Since $\pi{\circ}{\mu_1} = \mu_2$ and the definition~\eqref{eq:mu} is invariant:
\[
\pi{\circ}(\mu') = (\pi{\circ}{\mu})',
\]
we derive $\pi{\circ}{\mu'_1} = \mu'_2$.
\end{proof}

\section{Proofs for \cref{sec:lowerbound}}

\lemReduction*
\begin{proof}
    For the ``only if'' direction, if $\reachset M$ is finite
    then there is some $k$ s.t.~every reachable configuration $u$ has size $u(c_1)+u(c_2)+u(c_3)+u(c_4)+1 \leq k$,
    and thus the set of reversals of accepting runs can be recognised by a \kDTA {(k+1)},
    and thus also its complement can be recognised by a $(k+1)$-\DTA.
    
    For the ``if'' direction, if $\reachset M$ is infinite,
    then there exist reachable configurations with arbitrarily large counter values.
    Suppose, towards reaching contradiction, that $\lang A$ is recognised by a \kDTA k.
    Thus also its complement, that is the set of reversal-encodings of runs of $M$, is recognised by some \kDTA k $B$.
    There exists a run $\pi$ of $M$ where some counter value exceeds $k$,
    and thus when $B$ reads the reversal-encoding of $\pi$ it must forget some timestamp (say) $\tuple {c_1, t}$
    in some configuration $p_{i+1} \delta_{i+1} u_{i+1}$.
    Since $t$ is forgotten, we can perturb its corresponding $\tuple {c_1, t+1}$ in $p_i \delta_i u_i$
    to any value $\tuple {c_1, t'}$ \st $t' - t \neq 1$
    and obtain a new word still accepted by $A$,
    but which is no longer the reversal-encoding of a run of $M$,
    thus reaching the sought contradiction.
\end{proof}

\lemEasyUndec*
\begin{proof}
	We consider \kDTA k membership (the \kmDTA k m membership is treated similarly). 
	 Consider some fixed timed language $M \in \mathcal Y$ which is not recognised by any \kDTA{k}
	(relying on the assumption 3), over an alphabet $\Gamma$.
%
%
	For a given timed language $L \in {\mathcal Y}$, over an alphabet $\Sigma$, 
	we construct
	the following language over the extended alphabet $\Sigma \cup \Gamma \cup \{\$\}$:
	%
	\begin{align*}
		N \; := \; L \rhd \{\$\} \rhd \timedwords{\Gamma} \;\cup\; \timedwords{\Sigma} \rhd \{\$\} \rhd M \ \subseteq \ \timedwords{\Sigma \cup \Gamma \cup \{\$\}},
	\end{align*}
	where $\$ \not\in \Sigma \cup \Gamma$ is a fixed fresh alphabet symbol. 
	Since $\mathcal Y$ contains all timeless timed languages due to the assumption 1, and is closed under
	union and composition due to the assumption 2, the language $N$ belongs to $\mathcal Y$.
	
	\begin{claim*}
		$L=\timedwords{\Sigma}$  if, and only if, $N$ is recognised by a \kDTA{k}.
	\end{claim*}

	For the ``only if'' direction, if $L = \timedwords{\Sigma}$ 
	then clearly $N = \timedwords{\Sigma} \cdot \{\$\} \cdot \timedwords{\Gamma}$.
	Thus  $N$ is timeless and in consequence $N$ is recognised by a \kDTA{k}, as \kDTA{k} recognise all timeless
	timed languages for any $k\geq 0$.
	
	For the ``if'' direction suppose, towards reaching a contradiction,
	that $N$ is recognised by a \kDTA{k} $A$
	but $L\neq\timedwords{\Sigma}$.
	Assume, w.l.o.g., that $A$ is greedily resetting.
	Choose an arbitrary timed word $w = (a_1, t_1) \dots (a_n, t_n) \not\in L$ over $\Sigma$.
	Therefore, for any extension $v = (a_1, t_1) \dots (a_n, t_n) (\$, t_n + t)$ of $w$ by one letter, 
	we have
	\[  v^{-1} N = t + M = \setof{(b_1, t+u_1)\dots(b_m, t+u_m)}{(b_1, u_1)\dots(b_m, u_m) \in M}. \] 
	Choose $t$ larger than the largest absolute value $m$ of constants appearing in clock constraints in $A$, and
	let $(p, \mu)$ be the configuration reached by $A$ after reading $v$. 
	As $t > m$, all the clocks are reset by the last transition and hence $\mu(\x) = 0$ for all clocks $\x$.
	Consequently, if the initial control location of $A$ were moved to the location $p$, the so modified \kDTA{k} $A'$ would accept 
	the language $M$.
	But this contradicts our initial assumption that $M$ is not recognised by a \kDTA{k},
	thus finishing the proof.
\end{proof}

\end{document}